\documentclass[reqno]{amsart}

\usepackage{amssymb,amsmath,amsthm,amsfonts,mathtools,bbm,mathrsfs,enumitem,graphics,float,multirow}
\usepackage[normalem]{ulem}
\usepackage{hyperref}
\hypersetup{
    colorlinks,
    linkcolor={blue!80!black},
    citecolor={red!80!black},
    urlcolor={blue!80!black}
}
\usepackage{xargs}                      
\usepackage{xcolor}  
\usepackage[colorinlistoftodos,prependcaption,textsize=tiny]{todonotes}
\newcommandx{\unsure}[2][1=]{\todo[linecolor=red,backgroundcolor=red!25,bordercolor=red,#1]{#2}}
\newcommandx{\change}[2][1=]{\todo[linecolor=blue,backgroundcolor=orange!25,bordercolor=blue,#1]{#2}}
\newcommandx{\info}[2][1=]{\todo[linecolor=OliveGreen,backgroundcolor=OliveGreen!25,bordercolor=OliveGreen,#1]{#2}}

\usepackage[top=1.6in,bottom=1.5in,left=1.6in,right=1.6in]{geometry}
\allowdisplaybreaks
\usepackage{cellspace}
\newcommand\redout{\bgroup\markoverwith{\textcolor{red}{\rule[0.5ex]{2pt}{0.8pt}}}\ULon}

\newtheorem{theorem}{Theorem}
\newtheorem{lemma}[theorem]{Lemma}

\newtheorem{proposition}[theorem]{Proposition}

\newtheorem{corollary}[theorem]{Corollary}

\usepackage{tikz}
\usetikzlibrary{calc}

\tikzset{font= }

\newcommand\nc\newcommand
\nc\bfa{{\boldsymbol a}}\nc\bfA{{\boldsymbol A}}\nc\cA{{\mathscr A}}
\nc\bfb{{\boldsymbol b}}\nc\bfB{{\boldsymbol B}}\nc\cB{{\mathscr B}}
\nc\bfc{{\boldsymbol c}}\nc\bfC{{\boldsymbol C}}\nc\cC{{\mathscr C}}
\nc\bfd{{\boldsymbol d}}\nc\bfD{{\boldsymbol D}}\nc\cD{{\mathscr D}}
\nc\bfe{{\boldsymbol e}}\nc\bfE{{\boldsymbol E}}\nc\cE{{\mathscr E}}
\nc\bff{{\boldsymbol f}}\nc\bfF{{\boldsymbol F}}\nc\cF{{\mathscr F}}
\nc\bfg{{\boldsymbol g}}\nc\bfG{{\boldsymbol G}}\nc\cG{{\mathscr G}}
\nc\bfh{{\boldsymbol h}}\nc\bfH{{\boldsymbol H}}\nc\cH{{\mathscr H}}\nc\fH{{\mathfrak H}}
\nc\bfi{{\boldsymbol i}}\nc\bfI{{\boldsymbol I}}\nc\cI{{\mathcal I}}
\nc\bfj{{\boldsymbol j}}\nc\bfJ{{\boldsymbol J}}\nc\cJ{{\mathscr J}}
\nc\bfk{{\boldsymbol k}}\nc\bfK{{\boldsymbol K}}\nc\cK{{\mathscr K}}
\nc\bfl{{\boldsymbol l}}\nc\bfL{{\boldsymbol L}}\nc\cL{{\mathscr L}}
\nc\bfm{{\boldsymbol m}}\nc\bfM{{\boldsymbol M}}\nc\cM{{\mathscr M}}
\nc\bfn{{\boldsymbol n}}\nc\bfN{{\boldsymbol N}}\nc\sN{{\mathscr N}}
\nc\bfo{{\boldsymbol o}}\nc\bfO{{\boldsymbol O}}\nc\cO{{\mathscr O}}
\nc\bfp{{\boldsymbol p}}\nc\bfP{{\boldsymbol P}}\nc\cP{{\mathscr P}}
\nc\bfq{{\boldsymbol q}}\nc\bfQ{{\boldsymbol Q}}\nc\cQ{{\mathscr Q}}
\nc\bfr{{\boldsymbol r}}\nc\bfR{{\boldsymbol R}}\nc\cR{{\mathscr R}}
\nc\bfs{{\boldsymbol s}}\nc\bfS{{\boldsymbol S}}\nc\cS{{\mathscr S}}
\nc\bft{{\boldsymbol t}}\nc\bfT{{\boldsymbol T}}\nc\cT{{\mathscr T}}
\nc\bfu{{\boldsymbol u}}\nc\bfU{{\boldsymbol U}}\nc\cU{{\mathscr U}}
\nc\bfv{{\boldsymbol v}}\nc\bfV{{\boldsymbol V}}\nc\cV{{\mathscr V}}
\nc\bfw{{\boldsymbol w}}\nc\bfW{{\boldsymbol W}}\nc\cW{{\mathscr W}}
\nc\bfx{{\boldsymbol x}}\nc\bfX{{\boldsymbol X}}\nc\cX{{\mathscr X}}
\nc\bfy{{\boldsymbol y}}\nc\bfY{{\boldsymbol Y}}\nc\cY{{\mathscr Y}}
\nc\bfz{{\boldsymbol z}}\nc\bfZ{{\boldsymbol Z}}\nc\cZ{{\mathscr Z}}

\renewcommand{\le}{\leqslant}
\renewcommand{\leq}{\leqslant}
\renewcommand{\ge}{\geqslant}

\DeclareMathOperator{\rank}{rk}

\nc{\Cay}{{\sf Cay}}

\newcommand{\field}[1]{\mathbb{#1}}
\newcommand{\F}{\field{F}}
\nc{\ff}{{\mathbb F}}

\newcommand\remove[1]{}

\begin{document}
\title[]{Storage codes on coset graphs with asymptotically unit rate}
\author[]{Alexander Barg}\address{Department of ECE and Institute for Systems Research, University of Maryland, College Park, MD 20742, USA}\email{abarg@umd.edu}
\author[]{Moshe Schwartz}\address{School of Electrical and Computer Engineering, Ben-Gurion University of the Negev, Beer Sheva 8410501, Israel, and Department of Electrical and Computer Engineering, McMaster University, Hamilton, Ontario L8S 4K1, Canada}\email{schwartz.moshe@mcmaster.ca}
\author[]{Lev Yohananov}\address{School of Electrical and Computer Engineering, Ben-Gurion University of the Negev, Beer Sheva 8410501, Israel, and Institute for Systems Research, University of Maryland, College Park, MD 20742, USA}\email{levyuhananov@gmail.com}

\begin{abstract} A storage code on a graph $G$ is a set of assignments of symbols to the vertices such that every vertex can recover its value by looking at its neighbors. We consider the question of constructing large-size storage codes on triangle-free graphs constructed as coset graphs of binary linear codes. Previously it was shown that there are infinite families of binary storage codes on 
coset graphs with rate converging to 3/4. Here we show that codes on such graphs can attain rate asymptotically approaching 1.

Equivalently, this question can be phrased as a version of hat-guessing games on graphs (e.g., P.J. Cameron e.a., \emph{Electronic J. Comb.} 2016). In this language, we construct triangle-free graphs with success probability of the players approaching one as the number of vertices tends to infinity. 
Furthermore, finding linear index codes of rate approaching zero is also an equivalent problem.

Another family of storage codes on triangle-free graphs of rate approaching 1 was constructed earlier by A. Golovnev and I. Haviv (36th Computational Complexity Conf., 2021) relying on a different family of graphs.
\end{abstract}

\maketitle

\section{Introduction}
Suppose we are given a connected graph $G(V,E)$ on $N$ vertices. Denote by $\sN(v)$ the neighborhood of $v$ in $G$, i.e., the set of vertices of $G$ adjacent to $v$, and let $\ff_q$ be a finite field of size $q$. A \emph{storage code} for $G$ is a set $\cC$ of vectors $c=(c_v)_{v\in V}\in \ff_q^N$ such that for every $\bfc\in \cC$ and $v\in V$ the value of the coordinate $c_v$ is uniquely determined by the values of $c_w, w\in \sN(v)$.
 More formally, suppose that for every vertex $v$ there is a function $f_v:\ff_q^N\to \ff_q$ such that for every $\bfc\in \cC$ and every $v\in V$ we have $c_v=f_v((c_w)_{w\in \sN(v)}),$ where we assume some implicit fixed ordering of the vertices.

The concept of storage codes was introduced around 2014 in \cite{Mazumdar2015,Shanmugam2014} and studied in subsequent papers \cite{MazMcgVor2019,BargZemor2022}. 
In coding theory literature, this concept was motivated by a more general notion of codes with locality
\cite{gopalan2011locality}, which has enjoyed considerable attention during the last decade \cite{Ramkumar2022}.

The main problem associated with storage codes is constructing codes of large size, expressed
as a function of the parameters of the graph. To make the comparison easier, define the
\emph{code rate} $R_q(\cC,G):=\log_q(|\cC|)/N$, and let $R_q(G)=\max_{\cC} R_q(\cC,G)$. The first observation is that for dense graphs it is easy to construct codes of large size or rate: for instance, if the graph is complete, $G=K_N$, then the condition of vertex recovery from its neighbors will be satisfied if the codevectors satisfy a global parity check, i.e., $\sum_{v\in V} c_v=0$ for every $\bfc\in \cC$. In this case $|\cC|=q^{N-1}$ and $R(\cC)=R_q(K_N)=\frac{N-1}N$ for any $q$. A variant of this idea clearly applies if $G$ contains many cliques or large cliques, and therefore it makes sense to focus on the case of graphs with no cliques at all, i.e., triangle-free graphs. This case received considerable attention in the literature regarding both storage codes and guessing games on graphs which we discuss next.

The following version of guessing games on graphs, introduced in \cite{Riis2007}, turns out to be equivalent to the construction problem of storage codes. The vertices are assigned colors out of a finite set $Q$ of size $q,$ and each vertex attempts to guess its color based on the colors of its neighbors. The game is won if all the vertices correctly guess their colors. The strategy may be agreed upon before the start of the game, but no communication between the vertices is allowed once the colors have been assigned. Suppose that the assignment $x\in Q^N$ is chosen randomly from the available options, and the goal of the game is to minimize the probability of failure. The following strategy connects this game with storage codes. Let $\cC$ be a $q$-ary storage code for the graph. Every vertex $v$ assumes that the assignment is a codeword in $\cC$ and guesses its color from its neighbors, then the probability of success is $P_s(\cC,G)=\frac{|\cC|}{q^N}$. This quantity is sometimes expressed via the \emph{guessing number} of $G$, defined as $\text{\small\sf gn}_q(G)=N+\log_q\max P_s(\cC,G)$ \cite{Cameron2016}, so in our notation $\frac1N (\text{\small\sf gn}_q(G))=R_q(G)$.

It was soon realized that the problem of finding storage codes is also equivalent to constructing procedures for linear index coding with \emph{side information graph} $G$. We refer to the introduction of \cite{BargZemor2022} for a brief overview of the known results on storage codes as well as the connections to guessing and linear index coding, and to \cite{AK2018} for a more detailed presentation. 
In particular, as shown in \cite{AloLubStaWeiHas2008,Mazumdar2015}   
   $$
   1-I_q(G)\le R_q(G)\le 1-I_q(G)+N^{-1}\log_q(N\ln q),
   $$
where $I_q(G)$ is the rate of the linear index code on $G$.

This paper further develops the work started in \cite{BargZemor2022}, and it shares with it the general approach to the construction 
of storage codes. In particular, as in \cite{BargZemor2022}, we confine ourselves to {\em linear binary codes} and we assume that
the repair function of a vertex $f_v$ is simply a parity check. In other words, $c_v=\sum_{w\in \sN(v)}c_w$ holds for all $c\in \cC$ and
all $v\in V$, where the sum is computed modulo 2. This repair rule involves an implicit assumption that vertex recovery relies on the \emph{full parity}, i.e., all the neighbors contribute to the recovery (the general definition does not include this stipulation). Thus, the parities of the vertices are given by the corresponding rows of the matrix $A+I$, where $A=A(G)$ is the adjacency matrix of $G$ and $I$ is the identity (in other words, we are adding a self-loop to every $v\in V$). For a given graph $G$ define the \emph{augmented adjacency matrix} $\tilde A(G)\coloneqq A(G)+I$. The dimension of the code $\cC$ equals $\dim(\cC)=N-\rank \tilde A(G)$, and thus we are interested in constructing graphs for which $\tilde A(G)$ has the smallest possible rank given the number of vertices $N$.

The assumption of using full parities is essential for our results. To explain the reasons, recall first that an $N\times N$ matrix $M$ is said to fit a graph $G(V,E)$ if $M_{u,v}=0$ whenever the vertices $u$ and $v$ with $u\ne v$ are not adjacent, and $M_{u,u}\ne 0$ for all $u\in V.$ The minimum rank of $G$ over $\ff_q$ is defined as $\text{minrk}_q(G)=\min\{\rank_{\ff_q}(M)\mid M\in \ff_q^{N\times N}, M \text{ fits } G\}.$ It is clear that the matrix $M$ can be used to define a storage code on $G$: as long as there are no rows with just one nonzero element, we can take $M$ instead of $\tilde A$ as the parity-check matrix.  This implies that 
    $$
    R_q(G)\ge 1-\text{minrk}_q(G)/N,
    $$
and it is also known that $I_q(G)\le \text{minrk}_q(G)/N$ \cite{BarYossef2011}. 
The minimum rank of graphs was introduced by Haemers \cite{Haemers1981} in his work on Shannon capacity of graphs, and it was used later in studies of linear index codes 
\cite{LubetzkyStav2009} and complexity of arithmetic circuits \cite{Codenotti2000}. 
In their recent work on a conjecture from \cite{Codenotti2000}, Golovnev and Haviv 
\cite{GolovnevHaviv2021} estimated the minimum rank of generalized Kneser graphs (recall that the vertices of $K^<(n,s,m)$ are all $s$-subsets of $[n]$ and two sets are adjacent if they intersect on at most $m$ elements). The next theorem combines Lemma 19 and Theorem 22 from \cite{GolovnevHaviv2021} (see also \cite{Haviv2019}) and is rephrased to match our notation and terminology (\cite{GolovnevHaviv2021} did not make a connection to either storage codes or index codes or guessing games).
\begin{theorem}\label{thm:GH21} {\rm(\cite{GolovnevHaviv2021})} Let $r$ be a multiple of 6. Then the graph
$K^<(r,\frac r2,\frac r6)$ contains no triangles, and its minimum rank over $\ff_2$ is at most
$2^{r h(1/3)}<2^{0.92r},$ where $h(x)=-x \log_2x -(1-x)\log_2(1-x)$ is the entropy function. Consequently, there exists a sequence of storage codes $\cC_r$ of length $N=\binom r{r/2}$ and rate $R(\cC_r)\ge 1-\frac1{\sqrt{\pi r}}{2^{-0.08r}}.$
 \end{theorem}
The result of \cite{GolovnevHaviv2021} is in fact more general than stated above in several respects. Namely, the Kneser graphs $K^<(r,\frac r2,m)$ with $m\le n/(2l)$ do not contain odd cycles of any length up to $l$, the estimate of their minimum rank is valid for all finite fields, and by considering subgraphs of $K^<$ one can extend the main claim to code lengths other than those in the theorem. Related to this, it has been shown in \cite{Golovnev2018} that the minimum rank of general random graphs from the Erd{\H o}s--R{\'e}nyi ensemble behaves as $\Theta(N/\log N),$ i.e., is sublinear in $N$ over any fixed finite field.

We now return to the main topic of this paper. For a $d$-regular graph it is easy to construct a code of rate 1/2. This is a folklore result accomplished by placing bits on the edges and
assigning to every vertex the $d$-vector of bits written on the edges incident to it. 
The rate of the obtained code over the alphabet of size $2^d$ is readily seen to be 1/2.
Several codes of higher rates were constructed in \cite{Cameron2016}. Namely, its authors observed
that each of the following (regular, triangle-free) graphs: the Clebsch graph on 16 vertices, the Hoffman-Singleton graph
on 50 vertices, the Gewirtz graph on 56 vertices, and the Higman-Sims graph on 100 vertices yield storage codes of rate above 1/2; the highest rate of 0.77 attained by the last of these. See \cite[Prop.~9]{Cameron2016} for details.

In their search of higher-rate codes, the authors of \cite{BargZemor2022} studied triangle-free graphs formed as coset graphs of linear codes (see Sec.~\ref{sec:4} for their definition). They constructed an infinite family of graphs $G_r$ with $|V(G_r)|=2^{r+1}$ vertices, $r\ge 4$, that admit storage codes of rate 
   $
   R=\frac 34-2^{-r}.
   $
Moreover, \cite{BargZemor2022} also gave an example of a triangle-free graph on $N=2^{16}$ vertices 
with $\rank(\tilde A)=11818/65536$, i.e., code rate slightly above $0.8196$. The authors of 
\cite{BargZemor2022} speculated that there may exist families of triangle-free graphs that admit 
storage codes of rate approaching one. Without the assumption of full parities, an affirmative answer 
follows already from Theorem~\ref{thm:GH21}. As our main result, we show that this is also true if 
the code is constructed from the full adjacency matrix of its graph.

\begin{theorem} \label{thm:main} There exists an infinite family of connected triangle-free graphs $(G_m)_m$ on $N=N(m)$ vertices such that
$\frac 1N\rank(\tilde A(G_m))\to 0$ as $m\to\infty$. In other words, 
    $$
    \lim_{m\to\infty} R_2(G_m)=1
    $$
and 
    $$
    {\sf gn}_2(G_m)=N(1-o(1)).
    $$
\end{theorem}
If we let $R_2(N)=\max R_2(G_N)$ to be the largest code rate over all triangle-free graphs on $N$ vertices, then this theorem implies that 
$\limsup_{N\to\infty} R_2(N)=1$. Since $R_q(N)\ge R_2(N)$ for all alphabets with $q\ge 2$ a power of a prime, this claim is also true for all such $q$. Moreover, in the context of guessing games it has been shown that the storage capacity $R_q(G)$ is almost monotone in $q$. Namely, the following is true.
\begin{theorem}[\cite{Gadouleau2011,CM2011,Cameron2016}] For any graph $G$, alphabet size $q,$ and $\epsilon>0$
there exists $q_0(G,q,\epsilon)$ such that for all $q'>q_0$
   $$
   R_{q'}(G)\ge R_q(G)-\epsilon.
   $$
\end{theorem}

Graphs in the family $(G_m)_m$ are constructed as coset graphs of binary linear codes. Their explicit description will be given in  Sec.~\ref{sec:result} after we develop the tools for analyzing the rank of adjacency matrices of coset graphs. Along the way
we illustrate the usefulness of our approach by giving simple proofs of some of the earlier results concerning storage codes on 
coset graphs.

\section{The construction of \cite{BargZemor2022} and our approach}\label{sec:4}
To remind the reader the definition of coset graphs, suppose we are given a binary linear code 
$C\subset F^n$ with an $r\times n$ parity-check matrix $H$, where $F=\ff_2$. This means that $C=\ker(H)$ and $\dim(C)\ge n-r;$ in other words, $H$ is allowed to contain dependent rows. Denote by $S$ the set of columns of $H$. Now consider the Cayley graph $G=\Cay({F}^r,S)$ on the group $({F}^r)_+$ with generators in $S$,  where $r$ is the number of rows of $H$. The vertex set of $G$ is $V(G)={F}^r$, and a pair of vertices $v,v'$ are connected if there is a column (generator) $h\in S$ such that $v=v'+h$. Since the group is Abelian, the graph $G$ is undirected. The vertices of $G$ can be also viewed as cosets in $\F^n_2 /C$, with two cosets connected if and only if the Hamming distance between them (as subsets) is one. It is clear that $G$ is triangle-free once the minimum distance of the code $C$ is at least four. We will assume this throughout and note that in particular, this implies that the columns in $H$ are distinct.

Having constructed $G$ from $C$, we consider the code $\cC=\ker(\tilde A(G))$. Clearly, $\cC$ is a storage code for $G$. Below, to distinguish
between the seed codes $C$ and storage codes $\cC$, we call the former \emph{small codes} and the latter \emph{big codes}. 
Generally, finding the parameters of $\cC$ from the parameters of $C$ is a nonobvious task, and even if the small code
is simply a repetition code, computing the rank of $\tilde A$ is not immediate. This problem in fact was considered earlier in the context of constructing quantum codes \cite{Couvreur2013}, and we refer to \cite{BargZemor2022} for a brief discussion of this connection. The family of storage codes in \cite{BargZemor2022} is constructed starting with the parity-check matrix $H_r$ obtained from the parity-check matrix of the extended Hamming code of length $2^{r-1}$ (see Sec.~\ref{sec:result} for more details). Below we often find it convenient to include $0_r$ (the all-zero vector of length $r$) in the set $S$, and then the adjacency matrix $A(G)$ includes ones on the diagonal. Accordingly, in such cases we do not use the notation $\tilde A$ and write $A$ instead.

The approach taken in \cite{BargZemor2022} relies on detailed analysis of the action of the adjacency operator $A$ on the space
of functions $f:{F}^r\to{F}$. While it indeed yields the value of the rank for the graphs $G$ obtained from $H_r$, extending it
to other families of codes looks difficult. The approach taken here relies on the following observations. Generators $h\in S$ act on ${F}^r$ as permutations $\sigma_h$, and clearly $\sigma_h^2(v)=v$ for every $v\in {F}^r$. Therefore for a given $h\ne 0$, $\sigma_h$
can be written as a product of disjoint transpositions (cycles of length 2), 
    $
    \sigma_h=(v_0,v_1)(v_2,v_3)\dots(v_{2^{r}-2},v_{2^{r}-1}),
    $ 
where $v_{2i}+v_{2i+1}=h$ for all $i\ge 0$ (note that the labeling of the vectors depends on $h$). For any vector $v\in {F}^r$, not necessarily a generator, the action of $\sigma_v$ can be written as 
a $2^r \times 2^r$ permutation matrix, which we denote below by $\Gamma_v$. For every $v,u,w\in {F}^r$ we have $(\Gamma_v)_{u,w}=\mathbbm{1}(v=u+w)$, i.e., every row contains a single 1 in the column $u+v$. In particular, $\Gamma_0=I$. Denote by ${\cM}_r=\{\Gamma_v \mid v\in {F}^r\}$ the collection of matrices $\Gamma_v$ and note that it forms a multiplicative group 
$({\cM}_r)_\times\cong (F^r)_+$; thus, $\Gamma^2_v=I$ and $\Gamma_v\Gamma_w=\Gamma_w\Gamma_v$ for all $v,w$.

Our approach has a number of common features with \cite{BargZemor2022}, but it is phrased entirely in terms
of permutation matrices and their combinations, which supplant actions of the generators of the group. This view enables us to isolate subblocks of the adjacency matrix, and to manipulate those subblocks to rearrange the matrix so that it becomes possible to control the rank.

\section{Properties of permutation matrices}
In this section, we formulate some results about combinations (products and sums) of the matrices $\Gamma_v$, where the computations are performed modulo 2. Throughout we assume some fixed order of the vectors in $F^r$. 
Our first lemma expresses the following obvious fact: acting on $F^r$ first by $u$ and then by $v$ is the same as acting by $u+v$.
\begin{lemma}\label{lemma:p2s} 
	For any $v,w\in F^r$, $\Gamma_{{v}}\Gamma_{{w}} = \Gamma_{{v} + {w}}$.
\end{lemma}
\begin{proof}
First, $(\Gamma_v\Gamma_w)_{u,u'}=1$ if and only if $(\Gamma_v)_{u,z}=1$ and $(\Gamma_w)_{z,u'}=1$, i.e., $v=u+z$ and $w=z+u'$ for some $z\in F^r$, or $v+w=u+u'$.
On the other hand, by definition $(\Gamma_{v+w})_{u,u'}=1$ if and only if $v+w=u+u'$.
\end{proof}
\remove{\begin{corollary}\label{cor:2}
	For all $v,w\in F^r$,
	\begin{enumerate} 
		\item   $\Gamma^2_{{v}} = I$.\label{cor:2_1}
		\item $\Gamma_{{v}} \Gamma_{{w}} \in{M}_r$.\label{cor:2_2}
		\item $\Gamma_{{v}}\Gamma_{{w}} = \Gamma_{{w}}\Gamma_{{v}}$.\label{cor:2_3}
	\end{enumerate}
\end{corollary}}
For any $u,w\in F^r$ there exists exactly one $v\in F^r$ such that $(\Gamma_v)_{u,w}\ne 0$, and thus
$\sum_{v} \Gamma_v=1_{2^r\times 2^r}$ (the $2^r\times 2^r$ all-ones matrix). Thus, if a binary matrix $A$ can be written as a sum of permutations, this representation
is unique (up to an even number of repeated summands). Therefore, if $A_V:=\sum_{v\in V}\Gamma_v$, where $V\subseteq F^r$ is some subset, is a sum of permutation matrices, the quantity $s(A_V)\coloneqq|V|$ is well defined. Let $\cP_r\coloneqq\{A_V \mid V\subseteq F^r\}$ be the set of all sums of the permutation matrices. Below we usually suppress the subscript $V$ from the notation.
 
\begin{proposition} 
The set $\cP_r$ with operations $``+"$ and $``\cdot"$ forms a commutative matrix ring with identity. 
\end{proposition}
\begin{proof}
By definition, the matrices in $\cP_r$ form a commutative additive group. The remaining properties of the ring follow immediately from Lemma~\ref{lemma:p2s} which reduces sums of products of $\Gamma$'s to simply sums.
\end{proof}

As above, the rank in the next lemma is computed over $\ff_2$.
\begin{lemma} Let $A,B\in \cP_r$. 
	\begin{enumerate}\label{claim:4}
		\item If $s(A)$ is odd, then $A^2=I$ and hence $\rank(A)=2^r$. \label{claim:4_1}
		\item If both $s(A)$ and $s(B)$ are odd, then  $s(AB)$ is odd.\label{claim:4_2}
		\item If $s(A)$ is odd, then there is $C \in \cP_r$ such that $AC=CA=B$.\label{claim:4_3}
		\item For any $\Gamma\in\cM_r$, $\rank(A) = \rank(A\Gamma)$. \label{claim:4_5}
	\end{enumerate}
\end{lemma}

\begin{proof}		
(1) For odd $\ell$ let  $A = \sum^{\ell}_{i=1} \Gamma_{{v}_i}$ such that $\Gamma_{{v}_i}\in {\cM}_r$.
		Thus, 
		\begin{align*}
		A^2 & = \Big(\sum^{\ell}_{i=1} \Gamma_{{v}_i}\Big)^2 = \sum^{\ell}_{i=1} \Gamma^2_{{v}_i}  +  \sum_{i\neq j} \Gamma_{{v}_i+{v}_j}.
		\end{align*}
The last sum has every term appearing twice, so it vanishes mod 2, and the first is formed
of an odd number of identities $I$, so it equals $I$. Thus, $A$ is nonsingular.

(2) For $\ell$ and $m$ odd, let  $A = \sum^{\ell}_{i=1} \Gamma_{{v}_i},B =  \sum^{m}_{j=1} \Gamma_{{w}_j} \in \cP_r$ such that $\Gamma_{{v}_i},\Gamma_{{w}_j} \in {\cM}_r$.
		Therefore,
		\begin{align*}
		A  B &=  \sum^{\ell}_{i=1} \Gamma_{{v}_i}\cdot \sum^{m}_{j=1} \Gamma_{{w}_j} =  \sum^{\ell}_{i=1}\sum^{m}_{j=1} \Gamma_{{v}_i+{w}_j}.
		\end{align*}
Since $\ell$ and $m$ are odd, the number of terms in this sum is odd. If some of them are
repeated, they cancel in pairs, without affecting the parity.

(3) Let $C=AB$. Since $\cP_r$ is closed under multiplication, $C\in \cP_r$. Next, $CA=AC=A^2 B=B$.
 
(4) This is obvious since multiplying by a permutation matrix over any field preserves the rank.
	
\end{proof}

The next lemma establishes another simple property of the matrices in $\cP_r.$

\begin{lemma}\label{lemma:5} For $t\ge 2$, Let $A_1,A_2, \dots, A_t$ and $B$ be matrices in $\cP_r$ and suppose that $s(B)$ is odd.
Let $D = (A_1|A_2|\dots|A_t)$ and $D' = (A_1B|A_2B|\dots|A_tB)$, then $\mathrm{RowSpan}(D) = \mathrm{RowSpan}(D')$. 
\end{lemma}
\begin{proof}
Observe that
	\begin{align*}
		D' &= (A_1B|A_2B|\dots|A_tB) = (BA_1|BA_2|\dots|BA_t) = BD,\\
        D&=(A_1|A_2|\dots|A_t)=(BA_1B|BA_2B|\dots|BA_tB) = BD'.
	\end{align*}
The first of these relations implies that $\mathrm{RowSpan}(D')\subseteq \mathrm{RowSpan}(D)$, and the second implies the reverse inclusion.
	\end{proof}

\begin{lemma}
The action of $\cM_r$ on the set $\cP_r$ partitions this set into equivalence classes.
\end{lemma}
\begin{proof} Transitivity follows by the remark before Lemma \ref{lemma:p2s} (or from the lemma itself): if $A_1=\Gamma_u A_2$ and $A_2=\Gamma_v A_3$ then $A_1=\Gamma_w A_3$ for $w=u+v$.
\end{proof}

Denote by $[A]$ the equivalence class of $A\in \cP_r$ and note that, since the action of $\cM_r$ is faithful, $|[A]|=|\cM_r|=2^r$.

\begin{lemma}\label{lemma:7}
Let $D_i,i=1,\dots,t$ and $A$ be matrices from $\cP_r$ and let $D=(D_1|D_2|\dots|D_{t})$. If $D_i \in [A],i=1,\dots,t$ then 
$\rank(A) = \rank(D)$.
\end{lemma} 

\begin{proof} For simplicity, let $t=2$ and let $D_1=A\Gamma_u$ and $D_2=A\Gamma_v$ for some $u,v\in F^r$.
Clearly $\rank(A\Gamma_u|A\Gamma_v)=\rank(A\Gamma_u|A\Gamma_u)$ since $A\Gamma_u$ and $A\Gamma_v$ share the same column set.
Next, we have 
	\begin{align*}
	\rank(D)&= \rank(D_1|D_2) =\rank(A\Gamma_u|A\Gamma_v) \\
	&=\rank(A\Gamma_u|A\Gamma_u) =\rank(\Gamma_uA|\Gamma_uA)=\rank(\Gamma_u(A|A))=\rank(A).\qedhere
	\end{align*}
\end{proof}

\section{Partitioning the adjacency matrix}

Suppose we are given a small code $C\subset F^n$, and let $S=\{h_1,\dots,h_{n}\}$ be the set of columns of its parity-check matrix. We begin with an obvious claim, stated in the next lemma.
\begin{lemma}\label{lem8}
    Let $A$ be the adjacency matrix of the graph $G=\Cay(F^r,S)$, where $S$ may or may not include $0$, then
	$$
A = \Gamma_{h_1}+\Gamma_{h_2}+\dots+\Gamma_{h_{n}}.
	$$
\end{lemma}

Our aim at this point is to operate on submatrices of $A$, transforming it to a more convenient form while preserving the rank. With this
in mind, we will define permutations derived from the matrix $H$ but acting on subspaces of $F^r$.

For integers $s,t, 1\le s\le t\le r$ and a vector $x=(x_1,\dots,x_{r})\in F^r$ let $x^{(s,t)}=(x_s,x_{s+1},\dots,x_t)$.
Given $\ell\le r$ and a vector $u\in F^\ell$, let $S_u$ be the set
of suffixes of the columns whose prefix is $u$:
   \begin{equation}\label{eq:B}
     S_{u}=\{h^{(\ell+1,r)}\mid h\in S, h^{(1,\ell)}=u\}.
   \end{equation}
Finally, consider permutations $\Gamma_v,v\in F^{r-\ell}$ acting on $F^{r-\ell}$ by adding vectors from $S_u$. For a given $u$ define the matrix
     \begin{equation}\label{eq:C}
D_u  = \sum_{{v} \in S_u} \Gamma_{v}.
     \end{equation}
Assume by definition that if $S_u = \emptyset$ then $D_u=0_{2^{r-\ell}\times 2^{r-\ell}}$.

To give an example, let $r=3,\ell=1$, and take 
 $$H=\begin{bmatrix}0&1&0&1\\0&0&1&1\\0&0&0&1\end{bmatrix},$$ 
 then $S_0=\{00,10\},S_1=\{00,11\}$ and
    $$
    D_0=I+\begin{bmatrix}0&0&1&0\\0&0&0&1\\1&0&0&0\\0&1&0&0\end{bmatrix}, \; 
    D_1=I+\begin{bmatrix}0&0&0&1\\0&0&1&0\\0&1&0&0\\1&0&0&0\end{bmatrix},
    $$
    where we assumed that the rows and columns of the matrices are indexed by $F^2$ in lexicographical order.
    
Now let us look at how the matrices defined in \eqref{eq:C} relate to the adjacency matrix of the graph $G$.
For $\ell\in\{1,2,\dots,r\},$ let $F^\ell=\{{v}_0,{v}_1,\dots,{v}_{2^\ell-1}\}$ where the vectors are ordered lexicographically.
Starting with a small code $C$ with a fixed parity-check matrix $H$, let us form the matrices $D_{{v}_0},D_{{v}_1},\dots,D_{{v}_{2^\ell-1}}$.

\begin{lemma}\label{lem:9}
For any $0 \leq \ell \leq r$ the matrix $A$ can be written in the form
$$
A=\begin{array}{ c *{4}{@{\hspace*{0.3in}}{c}} }
 &\hspace*{.2in}v_0 &\hspace*{.3in}v_1&\hspace*{0.1in}\dots&\hspace*{-.2in}v_{2^\ell-1}\\[.05in]
v_0&\multicolumn{4}{c}{\multirow{4}{*}
{\hspace*{-.3in}$\left(\begin{array}{cccc}
D_{v_0+v_0}&D_{v_0+v_1}&\dots&D_{v_0+v_{2^\ell-1}}\\[.05in]
D_{v_1+v_0}&D_{v_1+v_1}&\dots&D_{v_1+v_{2^\ell-1}}\\[.05in]
\vdots&\vdots&\cdots&\vdots\\[.05in]
D_{v_{2^\ell-1}+v_0}&D_{v_{2^\ell-1}+v_1}&\dots&D_{v_{2^\ell-1}+v_{2^\ell-1}}
    \end{array}\right)$}
    }\\[.05in]
v_1\\[.05in]
\vdots\\[.05in]
v_{2^\ell-1}
\end{array},
$$

\vspace*{.1in}\noindent where the $2^{r-\ell}\times2^{r-\ell}$ blocks $D_{v_i+v_j}$ are defined in \eqref{eq:C}.
\end{lemma}
 
  \begin{proof}
Given $x,y\in F^r$, we have $A_{x,y}=\sum_{h \in S} {\mathbbm 1}(x+h=y)$. 
Let ${v_i} = x^{(1,\ell)}$ and ${v_j} = y^{(1,\ell)}$ be the $\ell$-prefixes of $x$ and $y$ and let $D_{v_i,v_j}$ be the
block in $A$ at the intersection of the stripes $v_i$ and $v_j$. Our goal is to show that $D_{v_i,v_j}=D_{v_i+v_j}$.

	Let $S_{{v}_i + {v}_j}$ be the set defined in~\eqref{eq:B}. 
Given $t,u\in F^{r-\ell}$, the element $$(D_{v_i,v_j})_{t,u}={\mathbbm 1}(t+h'=u),$$ where $h'\in S_{v_i+v_j}$ is an $(\ell-r)$ tail vector.
Rephrasing and using \eqref{eq:C},
	$$
	D_{v_i,v_j}  = \sum_{{h'} \in S_{{v}_i + {v}_j}} \Gamma_{{h}'} = D_{{v}_i + {v}_j}.
	$$
\end{proof}

\subsection{Zero-codeword-only codes and repetition codes}
As an example of using the above approach, consider the adjacency matrices of coset graphs of the zero-codeword-only codes and repetition codes. 
Their ranks are known \cite{Couvreur2013,BargZemor2022}, but we rederive them using the tools developed in the previous sections. 
Let $J=\Gamma_{1_r}$ be the matrix defined as $J_{u,v}={\mathbbm 1}(u+v={1_r})$, where ${1_r}$ denotes the all-ones vector of length $r$. Since we assumed that the vectors in $F^r$ are ordered lexicographically, $J$ is antidiagonal.  

\begin{proposition} \cite[Prop.9]{Couvreur2013} Let $C=\{0_r\}$ be the zero-codeword code with $S$ given by the standard basis $e_i,i=1,\dots,r$.
The adjacency matrix $A_r$ of the coset graph $\Cay(F^r,S)$ has rank $2^r$ if $r$ is even and $2^{r-1}$ if $r$ is odd.
\end{proposition}
\begin{proof} If $r$ is odd, then $A_r=\sum_{i=1}^r \Gamma_{e_i}$ has full rank by Lemma~\ref{claim:4}(1). Let us consider the case of $r$ even, writing the matrix $C$ as in Lemma~\ref{lem:9}. Let $\ell=1$ and note that $D_0=A_{r-1}$ and $D_1=I_{r-1}$, both of rank $r-1$.
Then
  $$
  A_r=\left(\!\begin{array}{c@{\hspace*{.05in}}c} A_{r-1} &I\\[.05in]I&A_{r-1}\end{array}\!\right).
  $$
Since $s (A_{r-1})$ is odd, Lemma~\ref{lemma:5} implies that we can multiply the upper stripe by $A_{r-1}$ block-by-block without
affecting the row space in this part, thus with no effect on the rank. Upon multiplying, we obtain
  $$
  \left(\!\begin{array}{c@{\hspace*{.05in}}c} I& A_{r-1}\\[.05in]I&A_{r-1}\end{array}\!\right),
  $$
which obviously is of rank $2^{r-1}$.  
\end{proof}
We could of course simply eliminate all the entries in one of the submatrices $A_{r-1}$ by row operations, but the above procedure models 
our approach in other constructions. We next exemplify it in a more complicated case of repetition codes.
Let $C'$ be the repetition code of length $r+1$ and redundancy $r$ defined by the parity-check matrix $H'= [I|1_r]$, 
where $I$ is the identity matrix of order $r$. Form the matrix $H_r=[H'|0_r]$ and consider the code
$C$ of length $r+2$ for which $H_r$ is a parity-check matrix. In the next lemma, we compute the rank of the adjacency matrix of the coset
graph of $C$. 
	\begin{proposition} Let $S=\{e_1,\dots,e_{r},1_r,0_r\}$.
		The adjacency matrix $A_r$ of the coset graph $\Cay(F^r,S)$ satisfies
	  $$
\rank(A_r) =
		\begin{cases*}
		2^r & if $r$ is odd\\
		 \frac{1}{2}(2^{r} - 2^{\frac{r}{2}}) & if $r$ is even.
		\end{cases*}
	  $$
	\end{proposition}
\begin{proof}
If $r$ is odd then $s(A_r)$ is odd, and by Lemma~\ref{claim:4}\eqref{claim:4_1} $A_r$ is a full-rank matrix, i.e., $\rank(A_r)=2^r$. 
For even $r$ we prove the result by induction on $r$. Take $r=2$, then $H_2=\Big[
\text{\small$\begin{array}{*{4}{@{\hspace*{0.05in}}c}}1&0&1&0\\0&1&1&0\end{array}$}\!\Big]$. Since $S=F^2$, the matrix $A_2$ is an all-ones matrix of rank 1, verifying the base case.

Let us assume that $\rank(A_{r-2})=\frac{1}{2}(2^{r-2} - 2^{\frac{r-2}{2}})$ and let us consider the matrix $A_r$. 
We decompose it into blocks taking $\ell=2$ in the construction of Lemma \ref{lem:9}. Observe that
   $$D_{(0,0)} = A_{r-2} + J, \;D_{(0,1)} = I, \;D_{(1,0)} = I \text{ and } D_{(1,1)} = J.$$
By Lemma~\ref{lem:9} we obtain 
   $$
		A_r= \left( \begin{array}{*{4}{@{\hspace*{-.05in}}c}}
		\mbox{  $A_{r-2} + J$} & \mbox{  $I$} & \mbox{  $I$} & \mbox{  $J$}\\ 
		\mbox{  $I$} & \mbox{  $A_{r-2} + J$} & \mbox{  $J$} & \mbox{  $I$}\\ 
		\mbox{  $I$} & \mbox{  $J$} & \mbox{  $A_{r-2} + J$}& \mbox{  $I$}\\ 
		\mbox{  $J$} & \mbox{  $I$} &\mbox{  $I$} & \mbox{  $A_{r-2} + J$}	
		\end{array} \right).
   $$
Next multiply the top horizontal stripe by $A_{r-2}+J$ and the bottom one by $J$, then we obtain the following matrix:
   \begin{align*}
		 \left( \begin{array}{*{4}{@{\hspace*{0.1in}}c}}
		\mbox{  $I$} & \mbox{  $A_{r-2} + J$}  & \mbox{  $A_{r-2} + J$}  & \mbox{  $A_{r-2}J + I$} \\ 
		\mbox{  $I$} & \mbox{  $A_{r-2} + J$} & \mbox{  $J$} & \mbox{  $I$}\\ 
		\mbox{  $I$} & \mbox{  $J$} & \mbox{  $A_{r-2} + J$}& \mbox{  $I$}\\ 
		\mbox{  $I$} & \mbox{  $J$} &\mbox{  $J$} & \mbox{  $A_{r-2}J + I$}	
		\end{array} \right).
		\end{align*}
Since $s(A_{r-2}+J)$ is odd, by Lemma~\ref{lemma:5} this has no effect on the rank. Next, let us eliminate the bottom stripe adding to it the first three stripes row-by-row, and after that cancel two matrices in each of the two middle stripes. Overall we obtain 
		\begin{align*}
		 \left( \begin{array}{*{4}{@{\hspace*{0.1in}}c}}
		\mbox{  $I$} & \mbox{  $A_{r-2} + J$}  & \mbox{  $A_{r-2} + J$}  & \mbox{  $A_{r-2}J + I$} \\ 
		\mbox{  ${0}$} & \mbox{  ${0}$} & \mbox{  $A_{r-2}$} & \mbox{  $A_{r-2}J$}\\ 
		\mbox{  ${0}$} & \mbox{  $A_{r-2}$} & \mbox{  $A_{r-2}$}& \mbox{  ${0}$}\\ 
		\mbox{  ${0}$} & \mbox{  ${0}$} &\mbox{  ${0}$} & \mbox{  ${0}$}	
		\end{array} \right).
		\end{align*}
Let $D_1 = (A_{r-2}|A_{r-2}J)$ and $D_2 = (A_{r-2}|A_{r-2})$. By Lemma~\ref{lemma:7}, $\rank(D_1)=\rank(D_2)=\rank(A_{r-2})$.
		Therefore, considering the first three stripes independently, we obtain
		\begin{align*}
		\rank(A_{r})
		 & = \frac{1}{4}\cdot 2^r + \frac{1}{2}\cdot(2^{r-2} - 2^{\frac{r-2}{2}}) + \frac{1}{2}\cdot(2^{r-2} - 2^{\frac{r-2}{2}})    \\
				& = \frac{1}{2} \cdot \Big( 2^{r} - 2^{\frac{r}{2}}  \Big).  \qedhere  
		\end{align*}
	\end{proof}

\section{A new code family}\label{sec:result}

The starting point of the construction is the binary Hamming code of length $2^{r}-1,r\ge 2$. Let $\fH_r$ denote its parity-check matrix written
in the standard form in which all the $r$-columns are ordered lexicographically with $0_{r-1}|1$ on the left and $1_r$ on the right. 
Since $\fH_r$ contains all the nonzero columns, its coset graph is a complete graph $K_{2^r}$. 
\begin{lemma}\label{lemma:H0}
Consider the matrix $\left(\fH_r^\intercal | 0_{(2^{r}-1)\times m}\right)^\intercal, m\ge 1$ 
and let  $S$ be the set of its columns. Let $G=\Cay(F^{r+m},S)$, then $\rank(\tilde A(G))=2^m$.
\end{lemma}
\begin{proof}
We give two proofs, of which the second paves the way for later results in this section.

There are $2^r$ vertices $v\in F^{r+m}$ that share a common $m$-suffix. They form a clique $K_{2^r}$, which is a connected component of the graph $G$. Thus $\tilde A(G)$ is a $2^m\times 2^m$ block-diagonal matrix with each block of rank 1, and so its rank is $2^m$.

Alternatively, with Lemma~\ref{lem:9} in mind, let $\ell=r$ and note that for any fixed $r$-prefix $u$ the set $S_u=\{0_m\}$. Thus the matrix $\tilde A(G)$ can be written as a $2^r\times 2^r$ block matrix with each block equal to $I_{2^m}$, confirming again that its rank is $2^m$.
\end{proof}

The approach taken in \cite{BargZemor2022}, as well as in our work, is to add rows and columns to $\fH_r$ in order to remove the codewords of weight 3 in the small code, while keeping the rank $\rank (\tilde A)$ low. The parity-check matrix $H_{s}$ of the small code is formed of ${s}$ rows at the top and some combinations of the matrices $\fH_r$ underneath them. The resulting big codes are denoted by $\cC_{{s},r}$. The construction is recursive, starting with ${s}=2$ (the base case) and adding one extra row to the top part in each step.
 
\subsection{The case of ${s}=2$}
Consider the following $(r+2)\times(2^r+2)$ parity-check matrix of the small code: 
  \begin{equation}\label{eq:BZ}
  H_2=\left(\begin{array}{*{7}{@{\hspace*{0.1in}}{c}}}
     0&0&0&\dots&0&1&1\\
     0&1&1&\dots&1&0&1\\\hline
     0& & &     & &0&0\\[-.05in]
     \vdots&\multicolumn{4}{c}{\fH_{r}}&\vdots&\vdots\\
     0& & &     & &0&0
     \end{array}
     \right).
   \end{equation}
By inspection, any three columns of this matrix not including the first column are linearly independent, so the coset graph of the code $\text{\rm ker}(H_2)$ is triangle-free.
The family of big codes $\cC_{2,r}$ obtained from $H_2$ was originally discovered in \cite{BargZemor2022} (up to a minor difference that is not essential for the results). As shown there, the rate of the storage code is $R(\cC_{2,r})=3/4-2^{-r}$ for all $r\ge 4$. 
While the argument in \cite{BargZemor2022} is somewhat complicated, here we give a straightforward proof based on the decomposition of Lemma \ref{lem:9} with $\ell=2$. \begin{proposition} Let $S$ be the set of columns of the matrix $H_2$ and let $G_2\coloneqq\Cay(F^r,S)$ be the coset graph
of the code $\text{\rm ker}(H_2)$. Then
$\rank(A(G_2))\le 2^{r+2}/4+4$.
\end{proposition}
\begin{proof}
	Let $P$ be the adjacency matrix of $K_{2^{r}}$. We partition the adjacency matrix of $G_2$ into blocks of order $2^{r}$ according
to the 2-prefix of the generator. The result can be written in the form
	\begin{align*}
	\renewcommand\arraystretch{1.2}
	A(G_2)  = \left( \begin{array}{*{4}{@{\hspace*{0in}}c}}
	\mbox{  $I$} & \mbox{  $P+I$} & \mbox{  $ I$} & \mbox{  $I$}\\
	\mbox{  $P+I$} & \mbox{  $I$} & \mbox{  $I$} & \mbox{  $I$}\\ 
	\mbox{  $I$} & \mbox{  $I$} & \mbox{  $I$}& \mbox{  $P+I$}\\ 
	\mbox{  $I$} & \mbox{  $I$} &\mbox{  $P+I$} & \mbox{  $I$}	
	\end{array} \right)=
	       \left(\begin{array}{*{4}{c}}
	       I&I&I&I\\I&I&I&I\\I&I&I&I\\I&I&I&I
	       \end{array}\right)
	+      \left(\begin{array}{*{4}{c}}
	0&P&0&0\\P&0&0&0\\0&0&0&P\\0&0&P&0
	  \end{array}\right),
	\end{align*}
 where the locations of $P$ correspond to $D_{01}$; see \eqref{eq:C} and
 Lemma \ref{lem:9}. Since $\rank(P) = 1$, the result follows.
\end{proof}

\begin{corollary}\label{cor:bz22}
The rate of the storage code is $$R(\cC_{2,r}) \ge 1 - \rank(A)/2^{r+2}=\frac 34-2^{-r}\to \frac 34$$ 
as $r\to\infty$. 
\end{corollary}

Next we increase the dimensions of the matrix We show this procedure in detail for ${s}=3$ and then state and prove the general claim.
 
\subsection{The case of {${s}=3$}} 
Form the matrix $(H_2|H_2)$ and 
  \begin{enumerate}
  \item add a new row $0_{2^{r}+2}1_{2^{r}+2}$ at the top of this matrix.
  \end{enumerate}
We obtain an $(r+3)\times(2^{r+1}+4)$ matrix of the form
    \begin{equation}\label{eq:BZ2}
  \left(\begin{array}{*{6}{@{\hspace*{0.1in}}{c}}c|}
     0&0&0&\dots&0&0&0\\
     0&0&0&\dots&0&1&1\\
     0&1&1&\dots&1&0&1\\\hline
     0& & &     & &0&0\\[-.05in]
     \vdots&\multicolumn{4}{c}{\fH_{r}}&\vdots&\vdots\\
     0& & &     & &0&0
     \end{array}
     \begin{array}{*{7}{@{\hspace*{0.1in}}{c}}}
     1&1&1&\dots&1&1&1\\
     0&0&0&\dots&0&1&1\\
     0&1&1&\dots&1&0&1\\\hline
     0& & &     & &0&0\\[-.05in]
     \vdots&\multicolumn{4}{c}{\fH_{r}}&\vdots&\vdots\\
     0& & &     & &0&0
     \end{array}
     \right).
   \end{equation}
This matrix contains linearly dependent triples that do not include the first column. To take care of them, 
\begin{enumerate}
\setcounter{enumi}{1}
\item add $r$ rows of zeros at the bottom,
\item replicate the column $1|0_{2r+2}$ immediately to the right of the vertical
divider $2^{r}-1$ times,
\item replace the $r\times(2^{r}-1)$ matrix of zeros at the bottom of the new columns with $\fH_{r}$. 
\end{enumerate}
The resulting $(2r+3)\times (3\cdot2^{r}+2)$ matrix $H_3$ is shown below in Eq.\eqref{eq:H3}, with boxes around the rows and columns formed in steps (2)-(4).

\begin{figure}[th]\begin{center}\scalebox{0.33}{\begin{tikzpicture}
\node[minimum size=0.5cm,scale=3] at (5.3,-9) 
{
\begin{minipage}{\textwidth}
\begin{equation}\label{eq:H3}
   H_3=\left(\begin{array}{@{\hspace*{0.1in}}c|*{4}{@{\hspace*{0.1in}}{c}}|@{\hspace*{0.1in}}c@{\hspace*{0.1in}}c
   *{4}{@{\hspace*{0.1in}}{c}}*{4}{@{\hspace*{0.1in}}{c}}|*{2}{@{\hspace*{0.1in}}{c}}}
   0&0&0&\dots&0&0&0&1&1&\dots&1 &1&1&\dots&1&1&1\\
   0&0&0&\dots&0&1&1&0&0&\dots&0 &0&0&\dots&0&1&1\\
   0&1&1&\dots&1&0&1&0&0&\dots&0 &1&1&\dots&1&0&1\\\hline
   0&&&& &0&0 &&&& &&&& &0&0\\ 
   \vdots &\multicolumn{4}{c|@{\hspace*{0.1in}}}{\fH_{r}} &\vdots&\vdots&\multicolumn{4}{c@{\hspace*{0.1in}}}{0}&\multicolumn{4}{c|@{\hspace*{0.1in}}}{\fH_{r}}&\vdots&\vdots\\
   0&&&& &0&0 &&&& &&&& &0&0\\    0&&&& &0&0 &&&& &&&& &0&0\\
   \vdots &\multicolumn{4}{c|@{\hspace*{0.1in}}}{0} &\vdots&\vdots&\multicolumn{4}{c@{\hspace*{0.1in}}}{\fH_{r}}&\multicolumn{4}{c|@{\hspace*{0.1in}}}{0}&\vdots&\vdots\\
   0&&&& &0&0 &&&& &&&& &0&0
   \end{array}\right).
 \end{equation}
   \end{minipage}
   };
   \draw[draw=blue] (4.8,-15.5) rectangle ++(5.8,13);
   \draw[draw=blue] (-6.4,-15.5) rectangle ++(26,4.4);
		\end{tikzpicture}}
 \end{center}
\end{figure}

In the next proposition, we establish simple properties of the matrix  $H_3$. Denote by $S$ the set of its columns.
\begin{proposition} The graph $G_3=\Cay(F^{2r+3},S)$ is connected and triangle-free. 
\end{proposition}
\begin{proof} Let $S_i,i=1,2,3$ be the sets of columns
that contain the first, second, and third Hamming matrices, respectively, and let
$S_4$ be the remaining set of 4 nonzero columns.
Let $0\in F^{2r+3}$ and $x\in F^{2r+3},x\ne 0$ be two group elements. To prove the first claim, it suffices to show that there is a path in $G_3$ that connects them. 
Suppose $x=(x_1,x_2,x_3,x_4,\dots,x_{2r+3})$. Since $S_4$ contains a basis of $F^3$, any assignment $x_1x_2x_3$ of the first three coordinates in $x$ can be reached independently of the remaining coordinates. Since the Hamming
matrices contain all the nonzero columns, it is also possible to reach any vector of the form $(000,x_4,\dots,x_{2r+3})$.

For the second claim, we need to show that all triples of columns that do not include the 
zero column are linearly independent. This is shown by a straightforward
case study. Let $h_1,h_2,h_3\in S\backslash \{0\}$ be three such columns and let $b=|\{h_1,h_2,h_3\}\cap S_4|$. If $b\in\{1,2,3\}$ then the claim is obvious.
If $b=0$, then $h_1,h_2,h_3$ can be chosen to intersect one, two, or all three
of the Hamming matrices. In each of these cases, direct inspection shows that the
triples cannot add to zero.
\end{proof}

\begin{proposition}  $\rank(A(G_3))\le2^{2r}+\frac 32\cdot2^{r+3}$.
\end{proposition}
\begin{proof}	
	Let $P_1$ be the adjacency matrix of the Cayley graph in $F^{2r}$ whose generators are the columns of the matrix $(\fH_{r}^\intercal |0_{(2^{r}-1)\times r})^\intercal$ and $0_{2r}$, and let $P_2$ be the same for the matrix $(0_{(2^{r}-1)\times r}|\fH_{r}^\intercal)^\intercal$ and $0_{2r}$. We
 arrange the matrix $A(G_3)$ in block form, where the blocks are indexed by binary 3-vectors (prefixes) ordered lexicographically. Using Lemma~\ref{lem:9} for the matrix $H_3$,
we obtain 
	\begin{align*}
	\renewcommand\arraystretch{2}
	A(G_3) =& \left( \begin{array}{*8{c@{\hspace*{0.1in}}}} 
	\mbox{  $I$} & \mbox{  $P_1+I$} &  \mbox{  $I$}  &  \mbox{  $I$}  &   \mbox{  $P_2+I$}  &  \mbox{  $P_1+I$}  &  \mbox{  $I$}  & \mbox{  $I$} \\ 
	\mbox{  $P_1+I$} & \mbox{  $I$}  &  \mbox{  $I$}  &  \mbox{  $I$}  &   \mbox{  $P_1+I$}  &  \mbox{  $P_2+I$}  &  \mbox{  $I$}  & \mbox{  $I$} \\ 
	\mbox{  $I$} & \mbox{  $I$} &  \mbox{  $I$}  &  \mbox{  $P_1+I$}  &   \mbox{  $I$}  &  \mbox{  $I$}  &   \mbox{  $P_2+I$}  &  \mbox{  $P_1+I$} \\ 
	\mbox{  $I$} & \mbox{  $I$} &  \mbox{  $P_1+I$}  &  \mbox{  $I$}  &   \mbox{  $I$}  &  \mbox{  $I$} &   \mbox{  $P_1+I$}  &  \mbox{  $P_2+I$}  \\ 
	 \mbox{  $P_2+I$}  & \mbox{  $P_1+I$} &  \mbox{  $I$}  &  \mbox{  $I$}  &   \mbox{  $I$} &  \mbox{  $P_1+I$}  &  \mbox{  $I$}  & \mbox{  $I$} \\ 
	 \mbox{  $P_1+I$} & \mbox{  $P_2+I$}  &  \mbox{  $I$}  &  \mbox{  $I$}  &   \mbox{  $P_1+I$}  &  \mbox{  $I$}  &  \mbox{  $I$}  & \mbox{  $I$} \\ 
	\mbox{  $I$} & \mbox{  $I$} &  \mbox{  $P_2+I$}   &   \mbox{  $P_1+I$}  &  \mbox{  $I$}  &  \mbox{  $I$}  &  \mbox{  $I$}  & \mbox{  $P_1+I$} \\ 
	\mbox{  $I$} & \mbox{  $I$} &  \mbox{  $P_1+I$} &   \mbox{  $P_2+I$}   &  \mbox{  $I$}  &  \mbox{  $I$}  & \mbox{  $P_1+I$}  & \mbox{  $I$} 
	\end{array} \right).
	\end{align*}
The dimensions of this matrix are $N\times N$, with $N=2^{2r+3}$, and each of the blocks is a square matrix of order $2^{2r}$. 
We can write $A(G_3)$ as a sum of three matrices,
   \begin{equation}\label{eq:BP}
   A(G_3)=A^{(1)}+A^{(2)}+A^{(3)},
   \end{equation}
where $A^{(1)}=(I)_{8\times8}$ is the matrix formed of identity blocks, $A^{(2)}=\begin{pmatrix}B|B\\ \hline B|B\end{pmatrix},$ and
  \begin{equation}\label{eq:BP1}
  B=\begin{pmatrix}0&P_1&0&0\\P_1&0&0&0\\0&0&0&P_1\\0&0&P_1&0\end{pmatrix}, \quad
   A^{(3)}=\begin{pmatrix} 0&0&0&0&P_2&0&0&0\\0&0&0&0&0&P_2&0&0\\
   0&0&0&0&0&0&P_2&0\\0&0&0&0&0&0&0&P_2\\
   P_2&0&0&0&0&0&0&0\\0&P_2&0&0&0&0&0&0\\
   0&0&P_2&0&0&0&0&0\\0&0&0&P_2&0&0&0&0
   \end{pmatrix}.
    \end{equation}
Now, Lemma~\ref{lemma:H0} says that 
$\rank(P_1)=\rank(P_2)=2^{r}$, and thus from \eqref{eq:BP}, \eqref{eq:BP1}, the rank of $A(G_3)$ is at most
	$$
	\rank(A(G_3)) \leq  2^{2r} + 4\cdot 2^{r} +8\cdot2^r= 2^{2r} + \frac32\cdot 2^{r+3}. 
	$$
\end{proof}
For obvious reasons, below we call matrices of the form $A^{(3)}$ \emph{block permutation matrices}.

\begin{corollary}
The rate of the storage code $\cC_{3,r}=\ker(A(G_3))$ is 
    $$
    R(\cC_{3,r}) \ge 1 - \frac{2^{2r} + \frac32\cdot 2^{r+3}}{2^{2r+3}}=\frac78-\frac32\cdot2^{-r}.
    $$
\end{corollary}

The general induction step is not different from the transition from ${s}=2$ to 3. Namely, we form the matrix $H_{s}$ by performing
steps (1)-(4) described above on the matrix $H_{{s}-1}$. As we will show shortly, this results in storage codes of rate
$(2^{s}-1)/2^{s}-o(1)$. 

Following the construction procedure, we find that the matrix $H_{s}$ is formed of ${s}$ horizontal stripes, ${s}-1$ of which
contain zero matrices and several Hamming matrices $\fH_{r}$. The set of columns is formed of $2^{{s}-1}-1$
vertical stripes each of which contains one Hamming matrix and $2^{{s}-1}+1$ other columns counting the all-zeros column. Thus the dimensions of $H_{s}$ are
  $$
  (({s}-1)r+{s})\times ((2^{{s}-1}-1)(2^{r}-1)+2^{{s}-1}+1).
  $$
Accordingly, the vertex set of the graph $G_{s}$ is of size $N=2^{({s}-1)r+{s}}$.
The $N\times N$ matrix $A(G_{s})$ can be written as a $2^{s}\times 2^{s}$ block matrix with each block of size $2^{({s}-1)r}$.
The structure of the matrix is similar to $A(G_2)$ and $A(G_3)$: namely, it is a sum of several $2^{s}\times 2^{s}$ block matrices, one of which is formed of identity matrices only. The remaining matrices are block permutation matrices similar to the third term in \eqref{eq:BP}, where in each of them, the block in question is the (augmented) adjacency matrix of the Cayley graph of the kind given in Lemma \ref{lemma:H0}.
Namely, the set of generators of this Cayley graph is the set of columns of the matrix 
     $$
    M_j\coloneqq\big[0_{(2^{r}-1)\times (jr)}\mid\fH_{r}^\intercal\mid 0_{(2^{r}-1)\times(({s}-2-j)r)}\big]^\intercal,
     $$
for some $j=0,1,\dots,{s}-2$ 
(thus overall the matrices $M_j$ are of dimensions $(({s}-1)r)\times (2^{r}-1)$\;).  Lemma \ref{lemma:H0} implies that for each
of these Cayley graphs, the adjacency matrix is of rank $2^{({s}-2)r}$.
Let us write the $N\times N$ matrix $A(G_s)$ as a sum similar to \eqref{eq:BP}, 
   $$
   A(G_s)=\sum_{j=1}^s A^{(j)},
   $$
where $A^{(1)}=(I)_{2^s\times 2^s}$ is a matrix formed of identity blocks and $A^{(s)}$ is a block
permutation matrix. Further, $\rank(A^{(s)})=2^{(s-2)r+s}$, and for each of the matrices $A^{s-j},j=1,2,\dots,s-2,$ the rank $\rank(A^{(s-j)})=\frac12\rank (A^{(s-j+1)}).$ 
Combining these data, we find that
    \begin{align}
     \rank(A(G_{s}))&\le 
    2^{(s-1)r}+
    \Big(1+\frac12+\dots+\frac1{2^{s-2}}\Big)2^{({s}-2)r+{s}}\nonumber\\
    &\le 2^{({s}-1)r}+2^{({s}-2)r+{s}+1}
    \nonumber\\
       &= N(2^{-{s}}+2^{-r+1}).
       \label{eq:rank-ell}
    \end{align}
We obtain the following result.
\begin{theorem}\label{thm:Rsr}
For any ${s}\ge 2$ there exists a sequence of storage codes $\cC_{{s},r}$ on triangle-free graphs on $N=2^{({s}-1)r+{s}}$ vertices, $r=4,5,\dots$
of rate 
  \begin{equation}\label{eq:Rlr}
   R(\cC_{{s},r})\ge 1-2^{-s}-2^{-r+1}.
   \end{equation}
\end{theorem}
\begin{proof}
Since $R=1-\rank(A(G_{s}))/N$, the claim about the rate follows from \eqref{eq:rank-ell}.
\end{proof}
	
We have constructed an infinite family of code sequences $\cC_{{s},r}, r=4,5,\dots; {s}=2,3,\dots$ whose rates are given by \eqref{eq:Rlr}. 
Taking $s,r\to\infty,$ we obtain a sequence of codes of length $N=N(s,r)$ whose rate converges to one.
This proves Theorem~\ref{thm:main}.
	
\subsection*{\sc Acknowledgments} We thank the editors for a careful reading of our manuscript and helpful suggestions. We are grateful to Ishay Haviv for drawing our attention to  papers \cite{Codenotti2000} and \cite{GolovnevHaviv2021}. Following the release of our preprint, Hexiang Huang and Qing Xiang informed us of their forthcoming work (since then published as \cite{HuangXiang-pub-2023}) on another generalization of the construction of \cite{BargZemor2022}, also attaining code rate one. This research was partially supported by NSF grants CCF2110113 (NSF-BSF) and CCF2104489, and by a German Israeli Project Cooperation (DIP) Grant under Grant PE2398/1-1.

\end{document}